\documentclass[12pt,reqno]{article}

\usepackage{xcolor}
\usepackage{amssymb}
\usepackage{amsmath}
\usepackage{amsthm}
\usepackage{amsfonts}
\usepackage{amscd}
\usepackage{graphicx}
\usepackage[shortlabels]{enumitem}

\usepackage[colorlinks=true,
linkcolor=webgreen,
filecolor=webbrown,
citecolor=webgreen]{hyperref}

\definecolor{webgreen}{rgb}{0,.5,0}
\definecolor{webbrown}{rgb}{.6,0,0}

\usepackage{color}
\usepackage{fullpage}
\usepackage{float}

\usepackage{graphics}
\usepackage{latexsym}
\usepackage{epsf}
\usepackage{breakurl}

\def\Enn{\mathbb{N}}

\def\suchthat{\, : \, }
\DeclareMathOperator{\theory}{FO}
\DeclareMathOperator{\faceq}{faceq}
\DeclareMathOperator{\appear}{appear}
\DeclareMathOperator{\leastappear}{leastappear}
\DeclareMathOperator{\recurfac}{recurfac}
\DeclareMathOperator{\recur}{recur}

\DeclareMathOperator{\leastrecur}{leastrecur}
\DeclareMathOperator{\spann}{span}
\DeclareMathOperator{\minspan}{minspan}
\DeclareMathOperator{\maxspan}{maxspan}
\DeclareMathOperator{\saspan}{saspan}

\DeclareMathOperator{\sa}{sa}

\usepackage{todonotes}

\newcommand{\recurconst}{\mathbf{R}}
\newcommand{\appearconst}{\mathbf{A}}

\begin{document}

\theoremstyle{plain}
\newtheorem{theorem}{Theorem}
\newtheorem{corollary}[theorem]{Corollary}
\newtheorem{lemma}[theorem]{Lemma}
\newtheorem{proposition}[theorem]{Proposition}

\theoremstyle{definition}
\newtheorem{definition}[theorem]{Definition}
\newtheorem{example}[theorem]{Example}
\newtheorem{conjecture}[theorem]{Conjecture}

\theoremstyle{remark}
\newtheorem{remark}[theorem]{Remark}

\title{String Attractors for Automatic Sequences}

\author{Luke Schaeffer\\
Institute for Quantum Computing \\
University of Waterloo \\
Waterloo, ON  N2L 3G1 \\
Canada \\
\href{mailto:lrschaeffer@gmail.com}{\tt lrschaeffer@gmail.com}
\and
Jeffrey Shallit\\
School of Computer Science \\
University of Waterloo \\
Waterloo, ON  N2L 3G1 \\
Canada \\
\href{mailto:shallit@uwaterloo.ca}{\tt shallit@uwaterloo.ca} }

\maketitle

\begin{abstract}
We show that it is decidable, given an automatic sequence $\bf s$ 
and a constant $c$, whether
all prefixes of $\bf s$ have a string attractor of size $\leq c$.
Using a decision procedure based on this result, we show that
all prefixes of the period-doubling sequence of length $\geq 2$
have a string attractor of size $2$.   We also prove analogous
results for other sequences, including the Thue-Morse sequence 
and the Tribonacci sequence.

We also provide general upper and lower bounds on string attractor size for different kinds of sequences.   For example, if $\bf s$ has a finite appearance constant, then there is a string attractor for
${\bf s}[0..n-1]$ of size $O(\log n)$.   If further
$\bf s$ is linearly recurrent, then there is a string
attractor for ${\bf s}[0..n-1]$ of size $O(1)$.  For automatic sequences, the size of the smallest string attractor for ${\bf s}[0..n-1]$ is either $\Theta(1)$
or $\Theta(\log n)$, and it is decidable which case
occurs.   Finally, we close with some remarks about greedy string attractors.
\end{abstract}

\section{Introduction}

Recently Kempa and Prezza \cite{Kempa&Prezza:2018}
introduced the notion of string attractor.
Let $w = w[0..n-1]$ be a finite word, indexed beginning at position $0$.
A {\it string attractor\/} of
$w$ is a subset $S \subseteq \{ 0,1,\ldots, n-1 \}$
such that every nonempty factor $f$ of $w$ has an occurrence
in $w$ that touches one of the indices of $S$.
For example, $\{2,3,4 \}$ is a string attractor for the English word
{\tt alfalfa}, and no smaller string attractor exists for that word.
Also see Kociumaka et al.~\cite{Kociumaka:2020}.

It is of interest to estimate the minimum possible size $\gamma(w)$ of a
string attractor for classically-interesting
words $w$.  This has been done,
for example, for the finite Thue-Morse words $\mu^n (0)$, where
$\mu$ is the map sending $0 \rightarrow 01$ and $1 \rightarrow 10$,
and $\mu^n$ denotes the $n$-fold composition of $\mu$ with itself.
Indeed, Mantaci et al.~\cite{Mantaci:2019}
proved that $\mu^n (0)$ has a string attractor of size
$n$ for $n \geq 3$,
and this was later improved to the constant $4$, for all $n \geq 4$, by 
Kutsukake et al.~\cite{Kutsukake:2020}.

In this note we first show that it is decidable, given an automatic
sequence $\bf s$ and a constant $c$,
whether all prefixes of $\bf s$ have a string attractor of size
at most $c$.   Furthermore, if this is the case, we can construct
an automaton that, for each $n$, provides the elements of a
minimal string attractor for the length-$n$ prefix of $\bf s$.
We illustrate our ideas by proving that the minimal string attractor
for length-$n$
prefixes of the period-doubling sequence is of cardinality $2$
for $n \geq 2$, and we obtain analogous results for the
Thue-Morse sequence, the Tribonacci sequence, and two others.

We use the notation $[i..j]$
for $\{ i, i+1, \ldots, i+j \}$, and ${\bf w}[i..j]$ for the factor $w[i] w[i+1] \cdots w[j]$.

\section{Automatic sequences}
\label{basics}

A numeration system represents each natural number $n$
uniquely as a word over some finite alphabet $\Sigma$.
If further the set of canonical representations is regular,
and the relation $z = x+y$ is recognizable by a finite automaton,
we say that the numeration system is {\it regular}.  
Examples of regular numeration systems include base $b$, for integers
$b \geq 2$ \cite{Allouche&Shallit:2003};
Fibonacci numeration \cite{Frougny:1986}; Tribonacci numeration
\cite{Mousavi&Shallit:2015}; and Ostrowski numeration systems
\cite{Baranwal:2020}.

Finally,
a sequence ${\bf s} = (s_n)_{n \geq 0}$ is {\it automatic\/} if
there exists a regular numeration system and an automaton
that, on input the representation of $n$, computes $s_n$.
We have the following result \cite{Charlier&Rampersad&Shallit:2012}:
\begin{theorem}
\label{thm1}
Let $\bf s$ be an automatic sequence.  
\begin{itemize}
\item[(a)] There is an algorithm that,
given a well-formed formula $\varphi$ in
the $\theory(\Enn, +, 0, 1, n \rightarrow {\bf s}[n])$
having no free variables, decides if $\varphi$ is true or false.
\item[(b)]
Furthermore, if $\varphi$ has free variables, then the algorithm
constructs an automaton recognizing the representation of the values
of those variables for which $\varphi$ evaluates to true.
\end{itemize}
\end{theorem}

We are now ready to prove
\begin{theorem}
\leavevmode
\begin{itemize}
\item[(a)]
It is decidable, given an automatic
sequence $\bf s$ and a constant $c$,
whether all prefixes of $\bf s$ have a string attractor of size
at most $c$.   
\item[(b)] Furthermore, if this is the case, we can construct
an automaton that, for each $n$, provides the elements of a
minimal string attractor for the length-$n$ prefix of $\bf s$.
\end{itemize}
\end{theorem}

\begin{proof}
\leavevmode
\begin{itemize}
\item[(a)]
From Theorem~\ref{thm1},
it suffices to create a first-order formula $\varphi$
asserting that the length-$n$ prefix of $\bf s$ has a string
attractor of size at most $c$.   We construct $\varphi$ in
several stages.

First, we need a formula asserting that the
factors ${\bf s}[i..j]$ and ${\bf s}[k+i-j..k]$ coincide.
We can do this as follows:
$$\faceq(i,j,k) := \forall u,v\ (u \geq i \wedge u\leq j \wedge v+j=u+k) 
\implies {\bf s}[u]={\bf s}[v]. $$

Next, let us create a formula asserting that
$\{ i_1, i_2, i_3, \ldots, i_c \}$ is
a string attractor for the length-$n$ prefix of
$\bf s$.  We can do this as follows:
\begin{align*}
\sa(i_1, i_2, \ldots, i_c, n)
& := (i_1<n) \wedge (i_2<n) \wedge \cdots \ \wedge (i_c<n) \\
& \forall k,l\ (k\leq l \wedge l<n) \implies \\
& (\exists r,s \ r\leq s \wedge s<n \wedge (s+k=r+l) \wedge \faceq(k,l,s)  \\
& \wedge ((r \leq i_1 \wedge i_1\leq s) \vee (r\leq i_2 \wedge i_2 \leq s)
	\vee \cdots \vee (r\leq i_c \wedge i_c \leq s))).
\end{align*}
Notice here that we do not demand that the $i_j$ be distinct, which explains
why this formula checks that the string attractor size is $\leq c$ and not
equal to $c$.

Finally, we can create a formula with no free variables asserting that
every prefix of $\bf s$ has a string attractor of cardinality
$\leq c$ as follows:
$$ \forall n \ \exists i_1, i_2, \ldots, i_c \ \sa(i_1,i_2, \ldots, i_c, n).$$

\item[(b)]
The algorithm in \cite{Charlier&Rampersad&Shallit:2012}
(which is essentially described in
\cite{Bruyere&Hansel&Michaux&Villemaire:1994})
constructs an automaton for $\sa$, which
recognizes $\sa(i_1, i_2, \ldots, i_c, n)$.   It is now easy to
find the lexicographically first set of indices $i_1, \ldots, i_c$
corresponding to any given $n$.
\end{itemize}
\end{proof}

One advantage to this approach is that it lets us (at least in principle)
compute the size of the smallest string attractor for {\it all\/} prefixes
of an automatic word, not just the ones of (for example) length $2^n$
(as in the case of the Thue-Morse words).  
Notice that knowing $\gamma$ on prefixes of length $2^n$ of
the Thue-Morse word $\bf t$ doesn't immediately give $\gamma$ on
all prefixes, since $\gamma$ need not be monotone increasing on prefixes.

For example, here is the size $s_n$ of the
smallest string attractor for length-$n$ prefixes of the
Thue-Morse infinite word $\bf t$, $1 \leq n \leq 32$:
\begin{table}[H]
\begin{center}
\begin{tabular}{c|cccccccccccccccc}
\hline
$n$ & 1& 2& 3& 4& 5& 6& 7& 8& 9&10&11&12&13&14&15&16\\
$s_n$ & 1& 2& 2& 2& 2& 2& 3& 3& 3& 3& 3& 3& 3& 3& 4& 4 \\
\hline
$n$ & 17&18&19&20&21&22&23&24&25&26&27&28&29&30&31&32\\
$s_n$ & 3& 3& 3& 3& 3& 3& 3& 3& 4& 4& 4& 4& 4& 4& 4& 4\\ 
\hline
\end{tabular}
\end{center}
\end{table}
\noindent Note that $\gamma({\bf t}[0..15]) = 4$, while
$\gamma({\bf t}[0..16]) = 3$.  We will obtain a closed form
for $\gamma({\bf t}[0..n-1])$ in Section~\ref{tm}.

\section{The period-doubling sequence}

Now let's apply these ideas to a famous 
$2$-automatic infinite word, the period-doubling sequence 
${\bf pd} = {\tt 101110101011} \cdots$
\cite{Guckenheimer:1977,Bellissard&Bovier&Ghez:1991,Damanik:2000}. It is the
fixed point of the morphism $1 \rightarrow 10$, $0 \rightarrow 11$.

We will use the theorem-proving software {\tt Walnut} \cite{Mousavi:2016}
created by Hamoon Mousavi; it implements the decision procedure
alluded to in Theorem~\ref{thm1}.   {\tt Walnut} allows the user
to enter the logical formulas we created above and determine the
results.   

Here is the {\tt Walnut} translation of our logical formulas.
In {\tt Walnut} the period-doubling sequence
is represented by the automaton named $\tt PD$.  The universal
quantifier $\forall$ is represented by {\tt A}, and the existential
quantifier $\exists$ is represented by {\tt E}.
The macro {\tt pdfaceq} tests whether
${\bf pd}[i..j] = {\bf pd}[k-j+i..k]$.
The rest of the
translation into {\tt Walnut} is probably self-explanatory.
\begin{verbatim}
def pdfaceq "k+i>=j & A u,v (u>=i & u<=j & v+j=u+k) => PD[u]=PD[v]":

def pdsa2 "(i1<n) & (i2<n) & Ak,l (k<=l & l<n) => (Er,s r<=s & s<n &
(s+k=r+l) & $pdfaceq(k,l,s) & ((r<=i1 & i1<=s) | (r<=i2 & i2<=s)))":

def pdfa1 "An (n>=2) => Ei1,i2 $pdsa2(i1,i2,n)":

def pdfa2 "$pdsa2(i1,i2,n) & i1 < i2":

def pdfa3 "$pdfa2(i1,i2,n) & (Ai3, i4 $pdfa2(i3,i4,n) => i4 >= i2)":

def pdfa4 "$pdfa3(i1,i2,n) & (Ai3 $pdfa3(i3,i2,n) => i3>=i1)":
\end{verbatim}

Here the {\tt Walnut} command {\tt pdfa1} returns {\tt true},
which confirms that every length-$n$
prefix of the period-doubling sequence, for $n \geq 2$,
has a string attractor of size at most $2$.  Since every word
containing two distinct letters requires a string attractor of
size at least $2$, we have now proved
\begin{theorem}
The minimum string attractor for ${\bf pd}[0..n-1]$ is $2$
for $n \geq 2$.
\end{theorem}

Achieving this was a nontrivial computation in {\tt Walnut}.
Approximately 45 gigabytes of storage and 60 minutes of CPU
time were required.
The number of states in each automaton, and the time
in milliseconds required to compute it, are given as follows:
\begin{table}[H]
\begin{center}
\begin{tabular}{l|r|l}
automaton & number & time \\
name & of states & in ms \\
\hline
{\tt pdfaceq} & 19 & 141 \\
{\tt pdsa2} & 20 & 3654862 \\
{\tt pdfa1} & 1 & 229 \\
{\tt pdfa2} & 14 & 2 \\
{\tt pdfa3} & 5 & 6 \\
{\tt pdfa4} & 5 & 3
\end{tabular}
\end{center}
\end{table}

Here is the automaton produced for {\tt pdfa4}.  It takes as
input the binary representation of triples of natural numbers,
in parallel.  If it accepts $(i_1, i_2, n)$, then
$\{ i_1, i_2 \}$ is a string attractor of size $2$ 
for the length-$n$ prefix of $\bf pd$.   
\begin{center}
\includegraphics[width=6in]{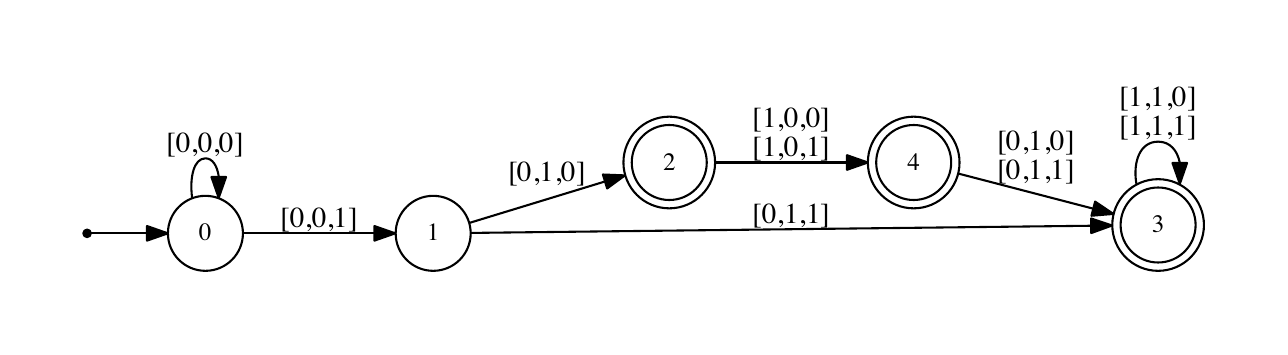}
\end{center}
For example, consider the prefix ${\bf pd}[0..25]$ of length $26$:
$10111010101110111011101010$.  In the automaton we see that
the base-$2$ representation of
$(7,15,26)$ is accepted, in parallel, so one string attractor is
$\{ 7, 15\}$:
$$1011101\underline{0}1011101\underline{1}1011101010.$$

By inspection of this automaton we easily get
\begin{theorem}
Let $n \geq 6$.   Then a string attractor for
${\bf pd}[0..n-1]$ is given by
$$
\begin{cases}
\{ 3 \cdot 2^{i-3} - 1, 3\cdot 2^{i-2} - 1 \}, & \text{if 
	$2^i \leq n < 3\cdot 2^i$;} \\
\{ 2^i -1, 2^{i+1} - 1 \}, & \text{if 
	$3 \cdot 2^i \leq n < 2^{i+1}$.}
\end{cases}
$$
\end{theorem}

\section{Back to Thue-Morse}
\label{tm}

We would like to perform the same kind of analysis for the Thue-Morse sequence, finding the size of the minimal string attractor for each of
its finite prefixes.  Although clearly achievable in theory,
unfortunately, a direct translation of what we
did for the period-doubling sequence fails to run to completion on
{\tt Walnut} within reasonable time and space bounds.  So we use
a different approach.

\begin{theorem}
Let $a_n$ denote the size of the smallest
string attractor for the length-$n$ prefix of the Thue-Morse
word $\bf t$.  Then 
$$
a_n = \begin{cases}
	1, & \text{if $n = 1$}; \\
	2, & \text{if $2 \leq n \leq 6$}; \\
	3, & \text{if $7 \leq n \leq 14$ or $17 \leq n \leq 24$}; \\
	4, & \text{if $n = 15, 16$ or $n \geq 25$}.
	\end{cases}
$$
\label{thue}
\end{theorem}

\begin{proof}
For $n \leq 59$ we can verify the values of $a_n = \gamma({\bf t}[0..n-1])$
with a machine computation.

Define $\rho_n (w)$ to be the number of distinct length-$n$ factors
of the word $w$, and set $\delta(w) = \max_{1 \leq i \leq n} \rho_i (w)/i$.
Christiansen et al.~\cite[Lemma 5.6]{Christiansen:2019} proved
that $\delta(w) \leq \gamma(w)$ for all words $w$.

Next, we observe that ${\bf t}[0..59]$,
the prefix of $\bf t$ of length $60$,
contains 40 factors of length $13$.   
Thus $\gamma(w) \geq \delta(w) \geq 40/13 > 3$ for all prefixes $w$
of $\bf t$ of length $\geq 60$.  

Next, we need to see that $\gamma(w) \leq 4$ for all prefixes $w$
of $\bf t$.   For prefixes of length $<12$ this can be verified
by a computation.   Otherwise, we claim that 
\begin{itemize}
\item[(a)]
$\{ 2^i - 1, 3\cdot 2^{i-1} - 1, 2^{i+1} - 1, 3\cdot 2^i - 1\} $
is a string attractor for ${\bf t}[0..n-1]$ 
for $13\cdot 2^{i-2} - 1 \leq n \leq 5 \cdot 2^i$ and $i \geq 2$;

\item[(b)]
$ \{ 3\cdot2^{i-1} -1, 2^{i+1} - 1, 3\cdot 2^i - 1, 2^{i+2} - 1\} $
 is a string attractor for ${\bf t}[0..n-1]$
for $9 \cdot 2^{i-1} - 1 \leq n \leq 3 \cdot 2^{i+1}$ and $i \geq 1$; 

\item[(c)]
$\{ 3 \cdot 2^{i-1} -1, 2^{i+1} -1, 2^{i+2} - 1, 5 \cdot 2^i - 1\}$
 is a string attractor for ${\bf t}[0..n-1]$
for $3 \cdot 2^{i+1} -1 \leq n \leq 13 \cdot 2^{i-1}$ and $i \geq 1$. \\
\end{itemize}
As can easily be checked, these intervals cover all $n \geq 12$.

To verify these three claims, we use {\tt Walnut} again.  
Here are the ideas, discussed in detail for the first of the three
claims.  Since we cannot express $2^i$ directly in {\tt Walnut},
instead we define a variable $x$ and demand that the base-$2$ representation
of $x$ be of the form $1 0^*$.  We can do this with the command
\begin{verbatim}
reg power2 msd_2 "0*10*":
\end{verbatim}
which defines a macro {\tt power2}.

Next, we define the analogue of {\tt pdfaceq} above; it tests
whether ${\bf t}[k..l] = {\bf t}[q-l+k..q]$.
\begin{verbatim}
def tmfaceq "(q+k>=l) & At Au (t>=k & t<=l & l+u=q+t) => T[t]=T[u]":
\end{verbatim}

Next we define four macros,
{\tt testa1, testa2, testa3, testa4} with
arguments $k,l,n,x$.
These test, respectively, if the factor
${\bf t}[k..l]$ has an occurrence in ${\bf t}[0..n-1]$
as ${\bf t}[p..q]$ such that 
\begin{enumerate}
\item $x - 1 \in [p..q]$; 
\item $3x/2 - 1 \in [p..q]$;
\item $2x - 1 \in [p..q]$;
\item $3x  - 1 \in [p..q]$.
\end{enumerate} 
Here $x = 2^i$.

{\footnotesize
\begin{verbatim}
def testa1 "Ep,q (k+q=l+p) & (p<=q) & (q<n) & $tmfaceq(k,l,q) & (p+1<=x) & (q+1>=x)":
def testa2 "Ep,q (k+q=l+p) & (p<=q) & (q<n) & $tmfaceq(k,l,q) & (2*p+2<=3*x) & (2*q+2>=3*x)":
def testa3 "Ep,q (k+q=l+p) & (p<=q) & (q<n) & $tmfaceq(k,l,q) & (p+1<=2*x) & (q+1>=2*x)":
def testa4 "Ep,q (k+q=l+p) & (p<=q) & (q<n) & $tmfaceq(k,l,q) & (p+1<=3*x) & (q+1>=3*x)":
\end{verbatim}
}
Finally, we evaluate the formula {\tt checka}, which
asserts that for all $n, x, k, l$ with 
$n \geq 12$, $x = 2^i$, $i \geq 2$,
$13 \cdot 2^{i-2} - 1 \leq n \leq 5 \cdot 2^i$,
and $k \leq l < n$, the factor ${\bf t}[k..l]$
has an occurrence in ${\bf t}[0..n-1]$
matching one of the four possible
values of the string attractor.
\begin{verbatim}
eval checka "An Ax Ak Al
((n>=12) & $power2(x) & (x>=4) & (13*x<=4*n+4) & (n <=5*x) & (k<=l) & (l<n)) =>
($testa1(k,l,n,x) | $testa2(k,l,n,x) | $testa3(k,l,n,x) | $testa4(k,l,n,x))":
\end{verbatim}
Then {\tt Walnut} evaluates this as true, so 
claim (a) above holds.

Claims (b) and (c) can be verified similarly.  We provide the {\tt 
Walnut} code below:

{\footnotesize
\begin{verbatim}
def testb1 "Ep,q (k+q=l+p) & (p<=q) & (q<n) & $tmfaceq(k,l,q) & (2*p+2<=3*x) & (2*q+2>=3*x)":
def testb2 "Ep,q (k+q=l+p) & (p<=q) & (q<n) & $tmfaceq(k,l,q) & (p+1<=2*x) & (q+1>=2*x)":
def testb3 "Ep,q (k+q=l+p) & (p<=q) & (q<n) & $tmfaceq(k,l,q) & (p+1<=3*x) & (q+1>=3*x)":
def testb4 "Ep,q (k+q=l+p) & (p<=q) & (q<n) & $tmfaceq(k,l,q) & (p+1<=4*x) & (q+1>=4*x)":

eval checkb "An Ax Ak Al
((n>=12) & $power2(x) & (x>=2) & (9*x<=2*n+2) & (n<=6*x) & (k<=l) & (l<n)) =>
($testb1(k,l,n,x) | $testb2(k,l,n,x) | $testb3(k,l,n,x) | $testb4(k,l,n,x))":

def testc1 "Ep,q (k+q=l+p) & (p<=q) & (q<n) & $tmfaceq(k,l,q) & (2*p+2<=3*x) & (2*q+2>=3*x)":
def testc2 "Ep,q (k+q=l+p) & (p<=q) & (q<n) & $tmfaceq(k,l,q) & (p+1<=2*x) & (q+1>=2*x)":
def testc3 "Ep,q (k+q=l+p) & (p<=q) & (q<n) & $tmfaceq(k,l,q) & (p+1<=4*x) & (q+1>=4*x)":
def testc4 "Ep,q (k+q=l+p) & (p<=q) & (q<n) & $tmfaceq(k,l,q) & (p+1<=5*x) & (q+1>=5*x)":

eval checkc "An Ax Ak Al
((n>=12) & $power2(x) & (x>=2) & (6*x<=n+1) & (2*n<=13*x) & (k<=l) & (l<n)) =>
($testc1(k,l,n,x) | $testc2(k,l,n,x) | $testc3(k,l,n,x) | $testc4(k,l,n,x))":
\end{verbatim}
}

\end{proof}

The reason why this approach
succeeded (and ran much faster than for the period-doubling
sequence) is that we could find a good guess as to what the string attractors
look like, and then {\tt Walnut} can be used to verify our guess.

\section{Going even further}

In this last section, we sketch the proof of two additional results,
which can be achieved in the same way we proved the result for the
Thue-Morse word.

The first concerns the Tribonacci word ${\bf TR} = 0102010 \cdots$,
the fixed point of the morphism sending $0 \rightarrow 01$,
$1 \rightarrow 02$, $2 \rightarrow 0$, which was studied,
for example, in
\cite{Chekhova&Hubert&Messaoudi:2001,Barcucci&Belanger&Brlek:2004,Rosema&Tijdeman:2005}.
Define the Tribonacci 
numbers, as usual, to be
$T_0 = 0$, $T_1 = 1$, $T_2 = 1$, and $T_n = T_{n-1} + T_{n-2} + T_{n-3}$
for $n \geq 3$.  Furthermore, define
$W_n = T_n + T_{n-3} + T_{n-6} + \cdots + T_{2+((n-2) \bmod 3)}$ for
$n \geq 4$.  

\begin{theorem}
The size of the smallest string attractor for the length-$n$ prefix
of the Tribonacci word is $3$ for $n \geq 4$.  Furthermore,
for $i \geq 4$ and $W_i \leq n < W_{i+1}$ a smallest string attractor of
size $3$ for ${\bf TR}[0..n-1]$ is $\{ T_{i-2}-1, T_{i-1}-1, T_i-1 \}$.
\end{theorem}

\begin{proof}
The proof proceeds almost exactly like the proof of Theorem~\ref{thue}.
Every prefix of length $n \geq 4$ contains the three distinct symbols
$0, 1, 2$, so clearly $\gamma({\bf TR}[0..n-1]) \geq 3$ for $n \geq 4$.

It remains to show the claim about the string attractor.  To do so
we create a {\tt Walnut} formula, given below.  We use the
so-called Tribonacci representation of numbers, described in
\cite{Bruckman:1989}.  It is easy to see that
the Tribonacci representation of $T_i$ is of the 
form $10^{i-2}$, and the Tribonacci representation of
$W_i$ is the length-$(i-1)$ prefix of the word $100100100 \cdots$.
We create regular expressions {\tt istrib} and {\tt t100} to match
these numbers.   The formula {\tt threetrib} accepts $x, y, z$
if $x = T_i, y = T_{i+1}, z = T_{i+2}$ for some $i \geq 2$.
The formula ${\tt adj}$ accepts $w_1, w_2$ if
$w_1 = W_i$, $w_2 = W_{i+1}$ for some $i \geq 4$.

{\footnotesize
\begin{verbatim}
reg istrib msd_trib "0*10*":
reg t100 msd_trib "0*(100)*100|0*(100)*1001|0*(100)*10010":

def tribfaceq "?msd_trib (q+k>=l) & At Au (t>=k & t<=l & l+u=q+t) => TR[t]=TR[u]":

def threetrib "?msd_trib $istrib(x) & $istrib(y) & $istrib(z) &
(x<y) & (y<z) & (Au (x<u & u<y) => ~$istrib(u)) & (Av (y<v & v<z) => ~$istrib(v))":

def adj "?msd_trib $t100(x) & $t100(y) & At (t>x & t<y) => ~$t100(t)":

def testtr "?msd_trib Ep,q (k+q=l+p) & (p<=q) & (q<n) & $tribfaceq(k,l,q) & (p<=w) & (q>=w)":

eval checktrib "?msd_trib An,x,y,z,w1,w2,k,l ((n>=4) & $threetrib(x,y,z) &
(x>=1) & (n>=w1) & (n+1<=w2) & $adj(w1,w2) & (z<=w1) & (2*z+x>=w2) & (k<=l) 
& (l<n)) => ($testtr(k,l,n,x-1) | $testtr(k,l,n,y-1) | $testtr(k,l,n,z-1))":
\end{verbatim}
}

Since {\tt checktrib} returns {\tt true}, the result is proven.
This was a large computation in {\tt Walnut}, requiring 20 minutes of CPU
time and 125 Gigs of storage.
\end{proof}

We now turn to another word, a variant of the Thue-Morse word
sometimes called {\bf vtm} and studied by Berstel \cite{Berstel:1979}.
It is the fixed point of the map $2 \rightarrow 210$, 
$1 \rightarrow 20$, and $0 \rightarrow 1$.

\begin{theorem}
The size of the smallest string attractor for the length-$n$ prefix
of the ternary infinite word {\tt vtm} is 
$$
\begin{cases}
1, & \text{if $n=1$;} \\
2, & \text{if $n=2$;} \\
3, & \text{if $3 \leq n \leq 6$;} \\
4, & \text{if $n \geq 7$.}
\end{cases}
$$
\end{theorem}

\begin{proof}
The result can easily be checked by a short computation for $n \leq 12$.

For a lower bound, observe that ${\bf vtm}[0..13]$ contains $10$ distinct
factors of length $3$, so $\gamma(w) \geq \delta(w) \geq 10/3 > 3$ for
all prefixes $w$ of $\bf vtm$ of length $\geq 14$.

We now claim that
\begin{itemize}
\item[(a)]
$\{ 2^i-1, 3\cdot 2^{i-1} - 1, 2^{i+1} - 1, 3 \cdot 2^i - 1 \}$ is a string
attractor for ${\bf vtm}[0..n-1]$ for
$13 \cdot 2^{i-2} \leq n < 5 \cdot 2^i$ and $i \geq 2$;

\item[(b)]
$\{ 2^i-1, 2^{i+1} - 1, 3 \cdot 2^i - 1, 9 \cdot 2^{i-1} - 1 \}$ is a string
attractor for ${\bf vtm}[0..n-1]$ for
$5 \cdot 2^i \leq n < 6 \cdot 2^i$ and $i \geq 2$;

\item[(c)]
$\{ 3 \cdot 2^{i-1}-1, 2^{i+1} - 1, 2^{i+2} - 1, 5 \cdot 2^i - 1 \}$
is a string attractor for ${\bf vtm}[0..n-1]$ for
$6 \cdot 2^i \leq n < 13 \cdot 2^{i-1}$ and $i \geq 2$.
\end{itemize}

This can be verified in exactly the same way that we verified the claims
for the Thue-Morse sequence $\bf t$.  The {\tt Walnut} code is given
below:

{\footnotesize
\begin{verbatim}
reg power2 msd_2 "0*10*":

def vtmfaceq "(q+k>=l) & At Au (t>=k & t<=l & l+u=q+t) => VTM[t]=VTM[u]":

def testva1 "Ep,q (k+q=l+p) & (p<=q) & (q<n) & $vtmfaceq(k,l,q) & (p+1<=x) & (q+1>=x)":
def testva2 "Ep,q (k+q=l+p) & (p<=q) & (q<n) & $vtmfaceq(k,l,q) & (2*p+2<=3*x) & (2*q+2>=3*x)":
def testva3 "Ep,q (k+q=l+p) & (p<=q) & (q<n) & $vtmfaceq(k,l,q) & (p+1<=2*x) & (q+1>=2*x)":
def testva4 "Ep,q (k+q=l+p) & (p<=q) & (q<n) & $vtmfaceq(k,l,q) & (p+1<=3*x) & (q+1>=3*x)":

eval checkva "An Ax Ak Al
((n>=13) & $power2(x) & (x>=4) & (13*x<=4*n) & (n<5*x) & (k<=l) & (l<n)) =>
($testva1(k,l,n,x) | $testva2(k,l,n,x) | $testva3(k,l,n,x) | $testva4(k,l,n,x))":

def testvb1 "Ep,q (k+q=l+p) & (p<=q) & (q<n) & $vtmfaceq(k,l,q) & (p+1<=x) & (q+1>=x)":
def testvb2 "Ep,q (k+q=l+p) & (p<=q) & (q<n) & $vtmfaceq(k,l,q) & (p+1<=2*x) & (q+1>=2*x)":
def testvb3 "Ep,q (k+q=l+p) & (p<=q) & (q<n) & $vtmfaceq(k,l,q) & (p+1<=3*x) & (q+1>=3*x)":
def testvb4 "Ep,q (k+q=l+p) & (p<=q) & (q<n) & $vtmfaceq(k,l,q) & (2*p+2<=9*x) & (2*q+2>=9*x)":

eval checkvb "An Ax Ak Al
((n>=20) & $power2(x) & (x>=4) & (5*x<=n) & (n<6*x) & (k<=l) & (l<n)) =>
($testvb1(k,l,n,x) | $testvb2(k,l,n,x) | $testvb3(k,l,n,x) | $testvb4(k,l,n,x))":

def testvc1 "Ep,q (k+q=l+p) & (p<=q) & (q<n) & $vtmfaceq(k,l,q) & (2*p+2<=3*x) & (2*q+2>=3*x)":
def testvc2 "Ep,q (k+q=l+p) & (p<=q) & (q<n) & $vtmfaceq(k,l,q) & (p+1<=2*x) & (q+1>=2*x)":
def testvc3 "Ep,q (k+q=l+p) & (p<=q) & (q<n) & $vtmfaceq(k,l,q) & (p+1<=4*x) & (q+1>=4*x)":
def testvc4 "Ep,q (k+q=l+p) & (p<=q) & (q<n) & $vtmfaceq(k,l,q) & (p+1<=5*x) & (q+1>=5*x)":

eval checkvc "An Ax Ak Al
((n>=24) & $power2(x) & (x>=4) & (6*x<=n) & (2*n<13*x) & (k<=l) & (l<n)) =>
($testvc1(k,l,n,x) | $testvc2(k,l,n,x) | $testvc3(k,l,n,x) | $testvc4(k,l,n,x))":
\end{verbatim}
}
The formulas {\tt checkva}, {\tt checkvb}, and {\tt checkvc} all return
{\tt true}, so the result is proved.
\end{proof}

\section{Non-constant string attractor size}

So far all the infinite words we have discussed have prefix
string attractor size bounded
by a constant.   However, if we can guess the form of a string attractor, then
we can verify it with a first-order formula, even for increasing size.

Let us consider the characteristic sequence of the powers of $2$:
${\bf p} = 11010001 \cdots $.   More precisely, ${\bf p}[n+1] = 1$ if
$n$ is a power of $2$, and $0$ otherwise.  This is an automatic sequence,
generated by a $3$-state automaton in base $2$.

\begin{theorem}
Suppose $n \geq 3$.  Then
\begin{itemize}
\item[(a)] $\{ 2 \cdot 4^j -1 \suchthat 0 \leq j \leq i \} \,\cup\, \{ 3 \cdot 4^i - 1 \}$
is a string attractor of cardinality $i+2$ for ${\bf p}[0..n-1]$ and 
$3 \cdot 4^i \leq n < 6 \cdot 4^i$ and $i \geq 0$;
\item[(b)]  $\{ 4^j -1 \suchthat 0 \leq j \leq i+1 \} \,\cup\, \{ 6\cdot 4^i - 1 \}$ is
a string attractor of cardinality $i+3$ for ${\bf p}[0..n-1]$ and $6 \cdot 4^i \leq n < 12\cdot 4^i$
 and $i \geq 0$.
\end{itemize}
\end{theorem}

\begin{proof}
Again, we use {\tt Walnut}.  Here {\tt P2} is the abbreviation for $\bf p$.
\begin{verbatim}
def p2faceq "(q+k>=l) & At Au (t>=k & t<=l & l+u=q+t) => P2[t]=P2[u]":
reg power4 msd_2 "0*1(00)*":

def test1 "An Ax Ak Al
($power4(x) & (3*x<=n) & (n<6*x) & (k<=l) & (l<n))
=> (Ep Eq Ey (k+q=l+p) & (p<=q) & (q<n) & $p2faceq(k,l,q) &
(($power4(y) & y<=x & p+1<=2*y & 2*y<=q+1) | (y=3*x & p+1<=y & y<=q+1)))":

def test2 "An Ax Ak Al
($power4(x) & (6*x<=n) & (n<12*x) & (k<=l) & (l<n))
=> (Ep Eq Ey (k+q=l+p) & (p<=q) & (q<n) & $p2faceq(k,l,q) &
(($power4(y) & y<=4*x & p+1<=y & y<=q+1) | (y=6*x & p+1<=y & y<=q+1)))":
\end{verbatim}
Both {\tt test1} and {\tt test2} return {\tt true}, so the theorem is proved.
\end{proof}

Of course, this theorem only gives an upper bound on the string attractor size
for prefixes of $\bf p$.  For a matching lower bound, other techniques must be used.

\section{Span}

After seeing a draft of this paper, Marinella Sciortino asked one of us about a certain feature of string attractors:   span.   The {\it span} of a string attractor $S = \{ i_1, i_2, \ldots, i_k \}$ with $i_1 < i_2 < \cdots < i_k $ is defined
to be $\spann(S) = i_k - i_1$.   It measures the distance between the largest and smallest elements of a string attractor.

One can then ask, among all string attractors of minimum cardinality for a finite string $x$, what are the maximum and
minimum span?   Call these quantities 
$\maxspan(x)$ and $\minspan(x)$.   We then have the following theorem:
\begin{theorem}
If ${\bf s}$ is an automatic sequence with $O(1)$ string attractor size, then $f(n) := \maxspan({\bf s}[0..n-1])$ and
$g(n) := \minspan({\bf s}[0..n-1])$ are synchronized
functions of $n$.    That is, there is an automaton that recognizes the representations of
$(n, f(n))$ and $(n,g(n))$.
\end{theorem}

\begin{proof}
It suffices to give a first-order formula for $\minspan$ and $\maxspan$.  Let $c$ be an upper bound on string attractor size for all 
prefixes of $\bf s$.   Then we can create formulas as follows:
\begin{align*}
\saspan(n,r) &:= \exists i_1, i_2, \ldots, i_c \ \sa(i_1, i_2, \ldots, i_c, n) \ \wedge\  (i_1 \leq i_2) \ \wedge\ (i_2 \leq i_3) \ \wedge\ \cdots \ \wedge (i_{c-1} \leq i_c) \\
& \wedge\ i_c = i_1 + r \\
\minspan(n,r) &:= \saspan(n,r) \ \wedge\ \forall s\ (s<r) \implies \neg\saspan(n,s) \\
\maxspan(n,r) &:= \saspan(n,r) \ \wedge\ \forall s\ (s>r) \implies \neg\saspan(n,s) .
\end{align*}
\end{proof}

As an example, we can carry this out for the period-doubling word with the following
{\tt Walnut} commands:
\begin{verbatim}
def pdsaspan "E i1, i2 $pdsa2(i1,i2,n) & i2=i1+r":
def pdsamind "$pdsaspan(n,r) & As (s<r) => ~$pdsaspan(n,s)":
def pdsamaxd "$pdsaspan(n,r) & As (s>r) => ~$pdsaspan(n,s)": 
\end{verbatim}
Looking at the results, we easily get
\begin{theorem}
For the period-doubling sequence the $\minspan$ of the length-$n$ prefix equals
$$
\begin{cases}
0, &\text{for $n = 1$}; \\
1, & \text{for $2 < n < 5$}; \\
2^i, & \text{for $3 \cdot 2^i \leq n < 3 \cdot 2^{i+1}$ and $i \geq 1 $ }.
\end{cases}
$$

The $\maxspan$ of the length-$n$ prefix equals
$$
\begin{cases}
0, &\text{for $n = 1$}; \\
1, & \text{for $2 \leq n \leq 3$}; \\
2^i, & \text{if $5 \cdot 2^{i-1} - 1 \leq n < 6 \cdot 2^{i-1} - 2$ and $i \geq 1 $} ; \\
3 \cdot 2^i, & \text{if $6 \cdot 2^i - 1 \leq n \leq 5 \cdot 2^{i+1} - 2$ and $i \geq 0$}.
\end{cases}
$$
\end{theorem}

\section{Asymptotic bounds on string attractor size}

We have seen several $k$-automatic sequences that have $O(1)$-size string attractors for all prefixes. Others, like the characteristic sequence of powers of $2$, have $\Theta(\log n)$-size string attractors for prefixes of length $n$. In this section, we show that these are the only two possibilities for $k$-automatic sequences: string attractors for length-$n$ prefixes have size $\Theta(1)$ or $\Theta(\log n)$ asymptotically.

We prove two upper bounds on the size of the string attractor in  Theorem~\ref{thm:appearanceattractor} and Theorem~\ref{thm:recurrenceattractor}. These do not require the word to be $k$-automatic, but instead have conditions on the \emph{appearance function} and \emph{recurrence function} (see, for example,
\cite[\S 10.9, 10.10]{Allouche&Shallit:2003}), which happen to be satisfied for $k$-automatic sequences, as discussed in Corollary~\ref{cor:automaticrecurrence}. 

Recall that an infinite word $\bf w$ is said to be \textit{recurrent} if every finite factor of $\bf w$ has infinitely many occurrences in $\bf w$. If, in addition, there is a constant $C$ such that two consecutive occurrences (or from the beginning of the word to the first occurrence) of every length-$n$ factor are separated by at most $C n$ positions, then $\bf w$ is said to be {\it linearly recurrent}. Let $\recurconst_\mathbf{w}$, the \emph{recurrence constant}, be the least such $C$. The related notion of the \emph{appearance function} records, for each $n$, the least smallest prefix that contains all length-$n$ factors. This, also, can be linearly bounded, i.e., there is a constant $C$ such that every length-$n$ appears in some prefix of at most $C n$ symbols. When it exists, we call the least such $C$ the \emph{appearance constant}, and denote it $\appearconst_\mathbf{w}$. 

%
%

The appearance and recurrence constants are computable for $k$-automatic sequences, so let us define the relevant first-order formulas:
\begin{itemize}
    \item $\faceq(i,j,n)$:  the factor 
    ${\bf w}[i..i+n-1]$ is the same as
    ${\bf w}[j..j+n-1]$
    \item $\appear(m,n)$:   all length-$n$ factors of $\bf w$ have an occurrence in $[0..m-1]$
    \item $\leastappear(m,n)$:  $m$ is the least integer such that $\appear(m,n)$ holds
    \item $\recurfac(i,n)$:   the factor
    ${\bf w}[i..i+n-1]$ occurs infinitely often
    in $\bf w$
    \item $\recur(m,n)$:  all length-$n$ factors of
    $\bf w$ have an occurrence contained in every
    length-$m$ factor
    \item $\leastrecur(m,n)$:  $m$ is the least
    integer such that $\recur(m,n)$ holds
\end{itemize}
Here are the definitions in first-order logic:
\begin{align*}
   \faceq(i,j,n) &:= \forall t \ (t<n) \implies 
    {\bf w}[i+t]={\bf w}[j+t] \\
   \appear(m,n) & := \forall i \ \exists j \ \faceq(i,j,n) \ \wedge\ j<m  \\
   \leastappear(m,n) &:= \appear(m,n) \ \wedge\ 
    \neg\appear(m-1,n) \\
    \recur(m,n) & := \forall i \ \forall k\ \exists j \ \faceq(i,j,n) \ \wedge\ k \leq j\ \wedge\ j+n\leq k+m \\
    \leastrecur(m,n) & := \recur(m,n) \ \wedge\ \neg\recur(m-1,n) 
\end{align*}

\begin{lemma}
    \label{lem:supremum}
    Let $S \subseteq \mathbb N_{>0} \times \mathbb N_{>0}$ be a $k$-automatic set such that for every $y \in \mathbb N_{>0}$, there are only finitely many $x \in \mathbb N_{>0}$ such that $(x,y) \in S$. Then $\sup_{(x,y) \in S} \frac{x}{y}$ is finite and computable.
\end{lemma}
\begin{proof}
    Briefly, $\frac{x}{y} \leq k^{s+1}$ for all $(x,y) \in S$ where $s$ is the number of states in the automaton accepting $S$. Otherwise, the base-$k$ representation of $x$ is longer than the representation of $y$ by more than a pumping length ($s$) and thus $(x,y) \in S$ can be pumped to infinitely many pairs $(x', y) \in S$, contradicting an assumption. This proves the supremum is finite; for computability, see our paper on critical exponents \cite{Schaeffer&Shallit:2012}. 
\end{proof}

\begin{corollary}
	\label{cor:automaticrecurrence}
	Let $\mathbf{w} \in \Sigma^{\omega}$ be $k$-automatic. Then $\appearconst_\mathbf{w}$ exists and is computable. If $\mathbf{w}$ is linearly recurrent, then $\recurconst_\mathbf{w}$ exists and is computable. 
\end{corollary}
\begin{proof}
    As we have seen, there are first-order predicates for appearance and recurrence. That is, predicates that accept $(m,n)$ if $m$ is the minimum length such that the window(s) (i.e., the prefix of length $m$, or every window of length $m$) intersects the desired factors (all or just recurrent factors) of length $n$. By the previous lemma, the appearance constant and recurrence constant are computable. 
\end{proof}

Let us specialize the definition of a string attractor to a limited set of lengths. 
\begin{definition}
	Given a word $\mathbf{w}$ and a set $L \subseteq \mathbb N$, a \emph{string attractor of $\mathbf{w}$ for lengths $L$} is a set of integers $S$ such that for every nonzero $\ell \in L$, every length-$\ell$ factor of $\mathbf{w}$ has some occurrence crossing an index in $S$. When $L$ is not specified, we take it to be $\mathbb N$, i.e., the string attractor property holds for all lengths. 
\end{definition}

\begin{lemma}
	\label{lem:appearance}
	Let $\mathbf{w} \in \Sigma^{\omega}$ be an infinite word with appearance constant $\appearconst_\mathbf{w} < \infty$. For all integers $n, s \geq 1$, there is a string attractor of $\mathbf{w}[0..n-1]$ for lengths $[s..2s]$ having size at most $2 \appearconst_\mathbf{w}$.
\end{lemma}
\begin{proof}
	All factors of length $\leq 2s$ in $\mathbf{w}$ intersect in a prefix $x$ of length at most $2\appearconst_\mathbf{w} s$ by the definition of the appearance constant. Take any arithmetic progression of stride $s$ through this interval, and observe that it has at most $2 \appearconst_\mathbf{w}$ terms.
	
	Obviously every factor of length $\geq s$ in $x$ intersects the arithmetic progression. On the other hand, all factors of length $\leq 2s$ have an occurrence intersecting the prefix. It follows that $S$ is a string attractor for $\mathbf{w}[0..n-1]$ for lengths $[s..2s]$.
\end{proof}

\begin{theorem}
	\label{thm:appearanceattractor}
	Let $\mathbf{w} \in \Sigma^{\omega}$ be an infinite word with appearance constant $\appearconst_\mathbf{w} < \infty$. There is a string attractor of $\mathbf{w}[0..n-1]$ of size $O(\appearconst_\mathbf{w} \log n)$.
\end{theorem}
\begin{proof}
	Cover the range of factor lengths ($1$ up to $n$) with intervals
	$$
	[1..2] \cup [3..6] \cup [7..14] \cup \cdots \cup [2^{k}-1..2^{k+1}-2]
	$$
	where the last interval contains $n$, so $k = O(\log n)$. By Lemma~\ref{lem:appearance}, there is a string attractor of $\mathbf{w}$ for the lengths in each interval of size $O(\appearconst_\mathbf{w})$. Define $S$ to be the union of these partial string attractors, and note that it is a string attractor of $\mathbf{w}$ for lengths $[1..n]$, and has size $O(\appearconst_\mathbf{w} \log n)$. 
\end{proof}
It turns out that this is tight, as shown by an example of Mantaci et al.\ \cite{Mantaci:2019} using the characteristic sequence of powers of two. There is a family of $\Omega(\log n)$ factors (e.g., factors of the form $10^{2^{2k}-1}1$) that have unique and pairwise disjoint occurrences in the prefix of length $n$, and thus every string attractor for this prefix must have at least $\Omega(\log n)$ indices. 

The key point in the argument above is that some factors do not recur often (or at all), requiring dedicated indices in the string attractor just for those factors. If every factor recurs sufficiently often, then indices meant to cover long factors will also be near various short factors. We may be able to move the index slightly so that it covers all the necessary long factors, but also some of the short factors. 

\begin{lemma}
	\label{lem:recurrence}
	Let $\mathbf{w} \in \Sigma^{\omega}$ be an infinite word with $\appearconst_\mathbf{w}, \recurconst_\mathbf{w} < \infty$. Suppose $S$ is a string attractor of $\mathbf{w}$ for lengths $L \subseteq [1..\lfloor \frac{s}{2 \recurconst_\mathbf{w}} \rfloor]$. For all $n, s \geq 1$, there is a string attractor of $\mathbf{w}[0..n-1]$ for lengths $L \cup [3s..6s]$ having size at most $\max(|S|, 3 \appearconst_\mathbf{w})$.
\end{lemma}
\begin{proof}
    Similar to Lemma~\ref{lem:appearance}, we cover the window containing all factors of length $\leq 6s$, (i.e., the prefix of length $\leq \min(n, 6 \appearconst_\mathbf{w} s)$ with an array of points where there is no gap of length $3s$ or longer without points, and thus the array intersects any factor of length $\geq 3s$. The difference is that we pack the points closer, averaging $\leq 2s$ from the next point, so that we may perturb each point left or right (independently of the other points) by up to $\frac{s}{2}$ without introducing a gap of $\geq 3s$. Let $S'$ be this set of points, with the perturbations of the elements to be determined later. Since the spacing of the points is initially $2s$, we need at most $3 \appearconst_\mathbf{w}$ points to cover the apperance window for factors of length $\leq 6s$.
    
    Recall that any window of length $s$ contains all factors of length $\leq \frac{s}{\recurconst_\mathbf{w}}$ because the first occurrence $\mathbf{w}$ is linearly recurrent with recurrent constant $\recurconst_\mathbf{w}$. In particular, all factors of length $\leq \frac{s}{2\recurconst_\mathbf{w}}$ that intersect some index $i \in S$ exist within an interval of length at most $\frac{s}{\recurconst_\mathbf{w}}$, and that factor is contained somewhere in any window of length $s$. Thus, we can perturb any index in $S'$ within its window of length $s$ to intersect all factors of length $\leq \frac{s}{2 \recurconst_\mathbf{w}}$ that some index $i \in S$ is responsible for in $S$. 
	
	Clearly the union $S \cup S'$ is a string attractor for lengths $L \cup [3s..6s]$ as desired, but it is too large. By the argument above, we can remove elements of $S$ by perturbing elements of $S'$, until we run out of indices in one or the other. If we run out of indices in $S$ then only indices of $S'$ are left and the string attractor has size $|S'| \leq 3 \appearconst_\mathbf{w}$. Otherwise, we removed one element of $S$ for each element of $S'$ we added, so the string attractor has size $|S|$. In either case, the size is at most $\max( |S|, 3 \appearconst_\mathbf{w})$, as required. 
\end{proof}

\begin{theorem}
	\label{thm:recurrenceattractor}
	Let $\mathbf{w} \in \Sigma^{\omega}$ be a linearly recurrent word with $\appearconst_\mathbf{w}, \recurconst_\mathbf{w} < \infty$. There is a string attractor of $\mathbf{w}[0..n-1]$ having size $O(\appearconst_\mathbf{w} \log \recurconst_\mathbf{w})$.
\end{theorem}
\begin{proof}
	Cover the interval of possible lengths, $[1,n]$, with $O(\log n)$ intervals as follows
	$$
	[1..2] \cup [3..6] \cup [6..12] \cup [12..24] \cup \cdots \cup [3 \cdot 2^k..6 \cdot 2^k].
	$$
	By Lemma~\ref{lem:appearance}, there is a string attractor of $\mathbf{w}$ for lengths $[1..2]$ having at most $2 \appearconst_\mathbf{w}$ elements. Thereafter, Lemma~\ref{lem:recurrence} says there is a string attractor for lengths $[3 \cdot 2^i..6 \cdot 2^k]$ of size $3 \appearconst_\mathbf{w}$, which also obsoletes every string attractor for lengths $[1.. \lfloor 2^{i-1} / \recurconst_\mathbf{w} \rfloor]$ having fewer than $3 \appearconst_\mathbf{w}$ elements.

	In other words, Lemma~\ref{lem:recurrence} says that the string attractor for lengths $[3 \cdot 2^i..6 \cdot 2^i]$ is no longer necessary once we have the string attractor for lengths $[3 \cdot 2^{i+k}..6 \cdot 2^{i+k}]$, where $k = \log_2 \recurconst_\mathbf{w} + O(1)$. It follows that the full string attractor $S$ for $\mathbf{w}[0..n-1]$ is the union of only the last $k$ string attractors (corresponding to the longest intervals), plus the string attractor for lengths $[1..2]$, and therefore 
	$$
	|S| \leq 3 k \appearconst_\mathbf{w} + 2\appearconst_\mathbf{w}  = O(\appearconst_\mathbf{w} \log \recurconst_\mathbf{w}).
	$$
\end{proof}

\begin{theorem}
	Let $\mathbf{w}$ be a $k$-automatic word. Let $\gamma_\mathbf{w}(n)$ be the size of the smallest string attractor of $\mathbf{w}[0..n-1]$. Then
	\begin{itemize}
		\item $\gamma_\mathbf{w}(n)$ is either $\Theta(1)$ or $\Theta(\log n)$,
		\item it is decidable whether $\gamma_\mathbf{w}(n)$ is $\Theta(1)$ or $\Theta(\log n)$, and
		\item in the case that $\gamma_\mathbf{w}(n) = \Theta(1)$, the sequence $\gamma_\mathbf{w}(1) \gamma_\mathbf{w}(2) \cdots$ is $k$-automatic and the automaton can be computed.
	\end{itemize} 
\end{theorem}
\begin{proof}
	By Corollary~\ref{cor:automaticrecurrence}, we know $\appearconst_\mathbf{w}$ and $\recurconst_\mathbf{w}$ exist and are computable. From the appearance constant alone, Theorem~\ref{thm:appearanceattractor} gives an $O(\log n)$ upper bound for $\gamma_{w}(n)$. 
		
	The set 
	\begin{align*}
	N := \{ i \in \mathbb{N} &: \text{for some $j>i$, $\mathbf{w}[i..j]$ has finitely many occurrences in $\bf w$} \\
	& \qquad \text{and $\mathbf{w}[i..j]$ is the first such} \}.
	\end{align*}
	is first-order expressible.
	Suppose $N$ is finite, with maximum element $a = \max N$. Then the suffix $\mathbf{w}[a+1..\infty]$ is recurrent and has $O(1)$ size string attractors by Theorem~\ref{thm:recurrenceattractor}, to which we add $\{ 1, \ldots, a \}$ to get $O(1)$ string attractors for every prefix of $\mathbf{w}$. In fact, we get a computable upper bound for the size of the string attractors: $\gamma_\mathbf{w}(n) \leq B$ where $B := a + O(\appearconst_\mathbf{w} \log \recurconst_\mathbf{w})$. We can use the formulas from Section~\ref{basics} to check the $\leq B$ element string attractors and determine the minimum size for all $n$. It follows that $\gamma_\mathbf{w}(1) \gamma_\mathbf{w}(2) \cdots$ is a $k$-automatic sequence. 
	
	Otherwise, $N$ is infinite. Hence $N$ is $k$-automatic, and so is the set $N' \subseteq N \times N$ where $(i,i') \in N'$ if $i' \in N$ is the first element of $N$ past the last occurrence of the shortest factor beginning at $i \in N$ that has finitely many occurrences. In other words, we might as well consider the \emph{shortest} non-recurrent factors starting at the positions in $N$, and when we do, $(i, i') \in N$ if $i'$ is next element of $N$ such that all occurrences of the factor at $i$ come before all occurrences of the factor at $i'$, and thus no occurrences overlap. 
	
	The fact that $i'$ is a $k$-synchronized function of $i$ (i.e., $(i,i')$ is $k$-automatic, and  for every $i$ there is only one $i'$) implies $i' = O(i)$, so there exists a constant $c$ such that $i' \leq ci$ for all $i \geq 1$. It follows that there is an infinite set of factors $\mathbf{w}[i_1..i_1+\ell_1-1], \mathbf{w}[i_2..i_2+\ell_2-1], \ldots$ where each $\mathbf{w}[i_j..i_j+\ell_j-1]$ occurs only in $\mathbf{w}[i_j..i_{j+1}-1]$ (and is therefore disjoint from occurrences of the others) and $i_j = O(c^j)$. The flip side is that $\Omega(\log n)$ of these factors have some occurrence in $\mathbf{w}[0..n-1]$, and must be intersected by $\Omega(\log n)$ distinct elements in the string attractor. Hence, $\gamma_\mathbf{w}(n) = \Omega(\log n)$, matching the earlier $O(\log n)$ upper bound.
\end{proof}

\section{Greedy string attractors}

Finding the optimal string attractor for a word is known to be $\mathsf{NP}$-complete \cite{Kempa&Prezza:2018}. Although this hardness result does not necessarily carry over to automatic sequences, it is still natural to consider approximations and heuristic algorithms. For example, we can define a string attractor greedily.
\begin{definition}
	Let $\mathbf{w} \in \Sigma^{\omega}$ be an infinite word. The \emph{greedy string attractor for $\mathbf{w}$}, $S$, is the limit of the sequence $\emptyset = S_0 \subseteq S_1 \subseteq \cdots \subseteq \mathbb N$, which is constructed iteratively as follows. For all $n \geq 0$, let $S_{n+1} = S_n \cup \{ j \}$ where $j$ is the smallest integer such that $S_{n}$ is \emph{not} a string attractor for $\mathbf{w}[0..j]$ (or $S_{n+1} = S_{n}$ if no such integer exists). 
\end{definition}

We say the first occurrence of a factor in $\mathbf{w}$ is \emph{novel}, and prove that the elements of the greedy string attractor are related to minimal novel factors. 
\begin{lemma}
	\label{lem:peaks}
	Suppose $S \subseteq \mathbb N$ is a string attractor for a prefix $\mathbf{w}[0..i]$ of an infinite word $\mathbf{w}$, and suppose $i \in S$. Let $j$ be the smallest integer such that $S$ is not a string attractor for $\mathbf{w}[0..j]$. Then $\mathbf{w}[j-\ell+1..j]$ is novel, no nonempty proper factor is novel. Also, $\ell$ is at most $j - i$. 
\end{lemma}
\begin{proof}
	Since $S$ is not a string attractor for $\mathbf{w}[0..j]$, there is some factor $\mathbf{w}[k..k+\ell-1]$ (let it be the shortest factor, i.e., take $\ell$ to be minimal) that has no occurrence in $\mathbf{w}[0..j]$ intersecting $S$. Nearly everything else follows because of a contradiction argument.
	\begin{itemize}
		\item The factor ends at position $j = k+\ell-1$ (i.e., is of the form $\mathbf{w}[j-\ell+1..j]$), otherwise $\mathbf{w}[0..j-1]$ is not a string attractor, contradicting the minimality of $j$.
		\item The factor is novel, since an earlier occurrence would either contradict that no occurrence intersects $S$ (if it intersected $S$), or the minimality of $j$ (if it did not intersect $S$). 
		\item The prefix $\mathbf{w}[j-\ell+1..j-1]$ is not novel because it would contradict the minimality of $j$. 
		\item The suffix $\mathbf{w}[j-\ell+2..j]$ is not novel because it would contradict the minimality of $\ell$. 
		\item Every other proper factor is contained in either $\mathbf{w}[j-\ell+1..j-1]$ or $\mathbf{w}[j-\ell+2..j]$ and therefore not novel.
	\end{itemize}
	Finally, since $\mathbf{w}[j-\ell+1..j]$ does not intersect $i$, we have $i \leq j-\ell \implies \ell \leq j-i$. 
\end{proof}

\begin{theorem}
	Let $S \subseteq \mathbb{N}$ be the greedy string attractor for a word $\mathbf{w} \in \Sigma^{\omega}$. If $\appearconst_\mathbf{w} < \infty$ then for all $n \geq 0$, $S \cap [0..n-1]$ is a string attractor for $\mathbf{w}[0..n-1]$ of size $O(\appearconst_\mathbf{w} \log n)$. 
\end{theorem}
\begin{proof}
	There is some iterate $S_k$ of the greedy construction, after which all indices added are $\geq n$. It is not hard to see that $S_k = S \cap \{0, \ldots, n-1\}$, and that $S_k$ is a string attractor for $\mathbf{w}[0..n-1]$ because otherwise the next step of the greedy construction would add an index $< n$. All that remains is to argue that $S_k$ has $O(\appearconst_\mathbf{w} \log n)$ elements. 
	
	By Lemma~\ref{lem:peaks}, two consecutive elements of $S$, $i$ and $j$, are separated by $\ell$, the length of a novel factor occurring after $i$. Since all factors of length $\leq i / \appearconst_\mathbf{w}$ have some occurrence beginning in the first $i$ symbols of $w$, a novel factor must be longer, i.e., $\ell > i / \appearconst_\mathbf{w}$. That is, 
	$$
	j \geq i + \ell \geq i +  \frac{i}{\appearconst_\mathbf{w}} \geq \left( 1 + \frac{1}{\appearconst_\mathbf{w}} \right) i.
	$$
	This guarantees at least geometric growth in the elements of $S$, and therefore there are at most logarithmically many elements in $S \cap [0..n-1]$. Since the logarithm is base $1 + \frac{1}{\appearconst_\mathbf{w}}$, and 
	$$
	\frac{1}{\log(1 + \frac{1}{\appearconst_\mathbf{w}})} \leq (\appearconst_\mathbf{w} + \tfrac{1}{2}), 
	$$
	the size of the string attractor is $O(\appearconst_\mathbf{w} \log n)$. 
\end{proof}

\section*{Acknowledgments}

We thank Jean-Paul Allouche for bibliographic suggestions.   We thank Marinella Sciortino for the suggestion to look at the minimum and maximum span of string attractors.

\end{document}